\documentclass[preprint]{article}

\usepackage{enumitem}
\usepackage{amsthm}
\usepackage{array}
\usepackage{alltt}
\usepackage{mathtools}
\usepackage{amssymb}
\usepackage{tikz}
\usetikzlibrary{automata,positioning}

\newcommand{\Bool}{\mathbb B}
\newcommand{\Nat}{\mathbb N}
\newcommand{\Lang}{\mathcal{L}}

\newcommand{\prop}{\mathit{Prop}}

\newcommand{\B}{\mathcal{B}}
\newcommand{\truths}{\mathcal{X}}

\renewcommand{\epsilon}{\varepsilon}
\renewcommand{\leq}{\leqslant}

\theoremstyle{definition}
\newtheorem{theorem}{Theorem}
\newtheorem{definition}[theorem]{Definition}
\newtheorem{lemma}[theorem]{Lemma}
\newtheorem{corollary}[theorem]{Corollary}
\title{Deciding minimal distinguishing DFAs is NP-complete}
\author{Jan Martens \\ j.j.m.martens@tue.nl}

\begin{document}
\maketitle

\begin{abstract}
    In this paper, we present a proof of the NP-completeness of computing the
    smallest Deterministic Finite Automaton (DFA) that distinguishes two given
    regular languages as DFAs. A distinguishing DFA is an automaton that
    recognizes a language which is a subset of exactly one of the given
    languages. We establish the NP-hardness of this decision problem by
    providing a reduction from the Boolean Satisfiability Problem (SAT) to
    deciding the existence of a distinguishing automaton of a specific size.
\end{abstract}

\section{Introduction}
We consider the problem of automatically explaining the inequivalence of
Deterministic Finite Automata (DFAs). In particular, we are interested in short
witnesses for the inequivalence. A straightforward approach to explain the
inequivalence of two DFAs would be to provide a distinguishing word, i.e. a word
that is accepted by one of the automata but not the other. 

This method of finding minimal distinguishing words is well understood and
decidable in polynomial time~\cite{moore1956gedanken}. An efficient
implementation is given in~\cite{smetsers2016minimal} that has the same runtime
complexity as the best known algorithm that decides language equivalence, known
as Hopcroft's minimization~\cite{hopcroft1971DFAmin}.

In this work we are motivated by smaller witnesses of inequivalence in the form
of regular languages. These languages might contain invariants that provide a
shorter and more intuitive explanation. For example, consider the DFAs
$\mathcal{A}$ and $\mathcal{B}$ shown in Figure~\ref{fig:DFAs}. The shortest
distinguishing word for these DFAs is $a^7$. Indeed, we confirm $a^7\in
\Lang(\mathcal{A})$ but $a^7\not\in \Lang(\mathcal{B})$. A different explanation
for the inequivalence of $\mathcal{A}$ and $\mathcal{B}$ could be: every odd
length sequence of $a$'s is accepted by $\mathcal{A}$ and not by $\mathcal{B}$. 

We call a DFA a \textit{distinguishing} automaton for two DFAs if the language
recognized is a subset of exactly one of the two DFAs. In the example from
Figure~\ref{fig:DFAs}, we see that our distinguishing witness with invariant is
equivalent to a distinguishing automaton with only two states, i.e. the DFA
$A_{odd}$ such that $\Lang(A_{odd}) = \{a^{2i + 1} \mid i\in \Nat\}$. An
automaton recognizing only the minimal distinguishing word $a^7$ would contain
at least eight states.

In the setting of model based development it can be key to understand the
differences between state based systems. This led us to study the synthesis of
\textit{distinguishing} DFAs, and leads naturally to following decision problem. 

\textbf{$k$-DFA-DIST:} Let $A_1$ and $A_2$ be DFAs such that $\Lang(A_1) \neq
\Lang(A_2)$, and $k\in\Nat$ a number. Decide if there is a DFA $A_{dist}$ with
at most $k$ states such that:
$$\Lang(A_{dist}) \subseteq \Lang(A_{1}) \iff  \Lang(A_{dist}) \not\subseteq \Lang(A_{2}).$$

\begin{figure}
    \centering
    \begin{tikzpicture}[node distance=1cm, initial text=$ $]
        \node[state, initial]  (q0) {$q_0$};
        \node[state, accepting, above right= of q0] (q1) {$q_1$};
        \node[state, accepting, below right= of q1] (q2) {$q_2$};
        \node[state, accepting, below left=of  q2] (q3) {$q_3$};
        
        \path[->] 
        (q0) edge  node[above left] {$a$} (q1)
        (q1) edge  node[above right] {$a$} (q2)
        (q2) edge  node[below right] {$a$} (q3)
        (q3) edge  node[below left] {$a$} (q0);


        \node[state, initial, right= 4.5cm of q0]  (p0) {$p_0$};
        \node[state, accepting, right= of p0] (p1) {$p_1$};
        \node[state, accepting, above right= of p1] (p2) {$p_2$};
        \node[state, accepting, below right=of  p2] (p3) {$p_3$};
        \node[state, below left= of p3] (p4) {$p_4$};
        

        \path[->] 
        (p0) edge  node[above] {$a$} (p1)
        (p1) edge  node[above left] {$a$} (p2)
        (p2) edge  node[above right] {$a$} (p3)
        (p3) edge  node[below right] {$a$} (p4)
        (p4) edge  node[left] {$a$} (p2);
    \end{tikzpicture}
    \caption{The DFA $\mathcal{A}$ on the left and the DFA $\mathcal{B}$ on the
    right side.\label{fig:DFAs}}
\end{figure}
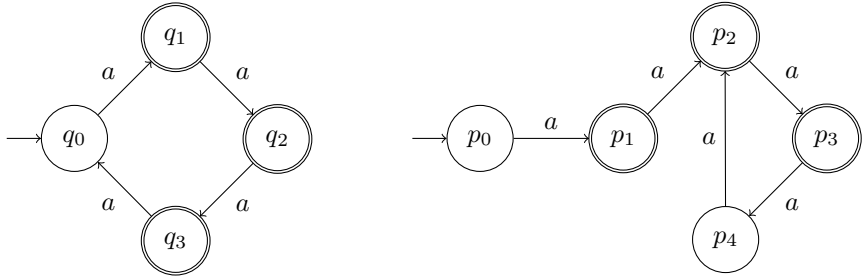

The contribution of this work is that we prove the intractability of
$k$-DFA-DIST. 
\begin{theorem}\label{thm:np-complete}
    Deciding $k$-DFA-DIST is NP-complete.
\end{theorem}

The reduction from CNF-SAT that proves the NP-completeness is new to our
knowledge. We believe this reduction of CNF-SAT formulas to regular languages is
an intuitive method of showing DFA problems NP-complete. 

There are some decision problems on DFAs that show some similarities, but are
different from the work here. For instance the early work of
Gold~\cite{gold1978} and Pfleeger~\cite{pfleeger1973state} in which it is shown
that learning minimal DFAs from (partial) observations is NP-complete. In the
line of this work by Gold, so-called \textit{separating} languages are widely
studied in the literature~\cite{demaine2011remarks,kupferman2022certifying}.
Here the separating problem is, given languages $L_1$ and $L_2$, to find a
separating language $L_{sep}$ such that $L_{sep} \subseteq L_1$ and $L_{sep}
\cap L_2 = \emptyset$. Although this resembles our distinguishing problem, a
direct relation is not trivial. 

Another influential work is due to Kozen~\cite{kozen1977lower}. This work
includes a proof of NP-hardness of deciding whether the intersection of a finite
number of DFAs is empty. 

\section{Notation \& Background}
For two natural numbers $i,j\in \Nat$ we write $[i,j] = {i, i+1, \dots , j}$ as
the closed interval from $i$ to $j$. Given a finite alphabet $\Sigma$, a
sequence of elements of $\Sigma$ is called a word. We define $\Sigma^i$ as the
set of all words over $\Sigma$ of length $i$, and $\Sigma^* = \bigcup_{i\in\Nat}
\Sigma^i$ for all words over $\Sigma$. Given words $u,v\in \Sigma^*$, we write
$u\cdot v$ and $uv$ for word concatenation. Additionally, given a number
$i\in\Nat$ and a word $u\in \Sigma^*$ we write $u^i$ for the concatenation of
$i$ times the word $u$.  

\begin{definition}
    A Deterministic Finite Automata (DFA) $A= (Q, \Sigma, \delta, q_0, F)$ is a
    five-tuple consisting of: 
    \begin{itemize}[itemsep=1pt,label=--]
        \item $Q$ a finite set of states, 
        \item $\Sigma$ a finite set of symbols called the alphabet, 
        \item $\delta: Q \times \Sigma \rightarrow Q$ the transition function,
        \item $q_0 \in Q$ the initial state, and
        \item $F \subseteq Q$ the set of final states.
    \end{itemize}
\end{definition}

The transition function $\delta$ extends naturally to a transition function for
words $\delta^*: Q \times \Sigma^* \rightarrow Q$. This is done inductively as
follows:
\begin{align*}
\delta^*(q,\epsilon) &= q\\
\delta^*(q, aw) &= \delta^*(\delta(q,a), w). 
\end{align*}

The language recognized by a DFA $A = (Q, \Sigma, \delta, q_0, F)$, is denoted
by $\Lang(A)$, and consists of all words $w\in \Sigma^*$ such that
$\delta^*(q_0, w) \in F$. 

The \textit{Myhill-Nerode theorem} is a useful tool to establish the number of
states necessary to recognize a language. It is based on the equivalence
relation relating words that have the exact same accepting extensions.

\begin{definition}
    Let $x,y\in \Sigma^*$ be words and $L\subseteq \Sigma^*$ a language, then $x
    \equiv_L y$ if and only if for all $z\in\Sigma^*$ it holds that $xz\in L
    \iff yz \in L$.
\end{definition}

\begin{theorem}{(Myhill-Nerode \cite[Theorem 3.9]{hopcroft1971DFAmin})} Let $L
    \subseteq \Sigma^*$ be a language, then $L$ is regular if and only if the
    relation $\equiv_L$ has a finite number of equivalence classes.
\end{theorem}

A more specific corollary of the theorem relates the number of equivalence
classes of $\equiv_L$ to the smallest number of states a DFA needs in order to
recognize $L$. 

\begin{corollary}\label{cor:myhill-nerode}
    Let $L$ be a regular language over an alphabet $\Sigma$, then the smallest
    DFA $A$ that recognizes $L$ has $k$ states where $k$ is the number of
    equivalence classes of the relation $\equiv_L$.  
\end{corollary}

\section{Reduction}
\newcommand{\1}{\mathtt{1}}
\newcommand{\0}{\mathtt{0}}
\renewcommand{\B}{\Bool}
Before we introduce the reduction we define some notation in which we encode
truth values of propositions. In the reduction we represent truth assignments as
words over the Boolean alphabet $\Bool = \{\0, \1\}$. Given a set of
propositional variables $\prop = \{p_1, \dots, p_k\}$, a truth assignment $\rho:
\prop \to \Bool$ is represented by the word $a_1 \dots a_k\in \B^k$, where $a_i
= \rho(p_i)$ for every $i\in [1,k]$. The set $\truths = \B^k$ defines all words
that represent truth assignments. 

Now we are ready to introduce our reduction from CNF-SAT in order to prove
Theorem~\ref{thm:np-complete}. Let $\phi = C_1 \wedge \dots \wedge C_n$ be a CNF
formula over the propositional variables $\prop = \{p_1, \dots, p_k\}$, we
define two regular languages over the alphabet $\Sigma = \Bool \cup \{\sharp\}$.
The first language $L^{-}_\phi \subseteq \Sigma^*$ is the finite set of at most
$n$ concatenated truth assignments separated by a $\sharp$, i.e.
\[
    L^{-}_\phi = \{w_1\sharp\dots w_j\sharp %
        \mid j \in [1,k]\text{ and } w_1, \dots ,w_j\in \truths\}.
\]

The second language $L^{+}_\phi \subseteq \Sigma^*$ is a superset of
$L^{-}_\phi$. In addition to all the word of $L^-_\phi$, the language
$L^{+}_\phi$ contains all words that have as prefix $n$ truth assignments $w_1,
\dots , w_n \in \truths$ that consecutively satisfy all clauses $C_1, \dots,
C_n$, more precisely that is, 
\begin{align*} L^{+}_\phi = L^{-}_\phi \cup \{w_1 \sharp \cdots w_n\sharp w \mid& w \in \Sigma^* \text{,} w_i\in \truths \text{ and $w_i$ satisfies } C_i \text{ for all } i\in[1,n]\}.
\end{align*}

The languages $L^{-}_\phi$ and $L^{+}_\phi$ are regular, and hence there are
automata that recognize these languages. In particular there are automata
recognizing these languages that are polynomial in size. One way of observing
this fact is by inspecting the number of Myhill-Nerode equivalence classes of
$L^{+}_\phi$ and $L^{-}_\phi$.

\begin{lemma}
    Given a CNF formula $\phi$, the languages $L_\phi^{+}$ and $L_\phi^{-}$ are
    recognizable by an automaton that is polynomial in the size of $\phi$.
\end{lemma}

The next lemma proves the key fact of our reduction. A truth assignment that
satisfies a CNF formula $\phi$ as recurring pattern forms a small distinguishing
automaton. Inversely a distinguishing automaton smaller than a certain size
necessarily implies the existence of a satisfying truth assignment for $\phi$.

\begin{lemma} \label{lemma:subseteq} Let $\phi = C_1 \wedge \dots \wedge C_n$ be
    a CNF formula over $k$ propositional letters $\prop = \{p_1, \dots, p_k\}$.
    Then $\phi$ is satisfiable if and only if there is a DFA $A_{dist}$ with at
    most $k+2$ states such that $\Lang(A_{dist}) \subseteq L^{+}_\phi$ and
    $\Lang(A_{dist}) \not\subseteq L^{-}_\phi$.
\end{lemma}
\begin{proof} 
    We prove both directions of the implication separately. 
    \begin{description}
        \item[($\Rightarrow$)] Assume $\phi$ is satisfiable, then there is a satisfying truth
assignment $\rho$ that is mapped to the word $w_\rho = \rho(p_1)\dots\rho(p_k)
\in \truths$. We define the language $L_{dist} = \{(w_\rho \cdot \sharp)^i \mid
i \in \Nat \}$, and show that $L_{dist}$ witnesses this implication. 

First we show that $L_{dist} \subseteq L^{+}_\phi$. Assume $i\in \Nat$, if
$i\leq n$ then by definition $(w_\rho\cdot \sharp)^i \in L_{\phi}^-$ and hence
also in $(w_\rho\cdot \sharp)^i \in L_{\phi}^+$. If $i > n$, since $\rho$ is a
satisfying assignment, it holds for any $w'\in \Sigma^*$ that $(w_\rho\cdot
\sharp)^n w' \in L^{+}_\phi$, and thus also $(w_\rho\cdot \sharp)^n (w_\rho\cdot
\sharp)^{i-n} \in L^{+}_\phi$. By covering both cases this means $L_{dist}
\subseteq L^{+}_\phi$. 

Next, we observe that $(w_\rho\cdot \sharp)^{n+1}\not \in L^{-}_\phi$, and thus
$L_{dist}\not\subseteq L^-_\phi$. Hence, since $L_{dist} \subseteq L^{+}_\phi$
any DFA that recognizes $L_{dist}$ is a distinguishing automaton. 

The minimal DFA $A_{dist}$ such that $\Lang(A_{dist}) = L_{dist}$ contains one
loop with $k+1$ states containing all positions of the word $w_\rho\cdot \sharp$
and a sink state to reject all other words. Thus, if $\phi$ is satisfiable we
can construct $A_{dist}$ with $k+2$ states that distinguishes $L^{+}_\phi$ and
$L^{-}_\phi$, which was to be showed. 

\item[($\Leftarrow$)] We assume $A_{dist}$ is a DFA with at most $k+2$ states
such that for the language accepted $\hat{L} = \Lang(A_{dist})$ it holds that
$\hat{L}\subseteq L_\phi^{+}$ and $\hat{L} \not\subseteq L_\phi^{-}$. We show
that this means $\phi$ is satisfiable.

Since $\hat{L} \setminus L_{\phi}^- \neq \emptyset$ and $\hat{L} \subseteq
 L_\phi^+$ there is a word $w\in L^+_\phi \setminus L^-_\phi$ accepted by
 $A_{dist}$. By definition $w$ is of shape $w = w_1 \sharp \dots w_n\sharp w'$
 where $w'\in \Sigma^+$ and $w_1,\dots, w_n\in \truths$ and for every $i\in
 [1,k]$ the word $w_i$ represents a satisfying truth assignment for the clause
 $C_i$. Next we show that $w_1$ represents a satisfying truth assignment for
 $\phi$ by counting the number of equivalence classes of $\equiv_{\hat{L}}$ for
 the prefixes of $w_1\cdot \sharp$, together with the postfix $w_{post} = w_2
 \sharp \dots w_n\sharp w'$ that witnesses an accepting postfix for $w_1\sharp$. 

 We define the set $U$ as the set containing all prefixes of $w_1=a_1\dots a_k$,
 i.e. 
 \[
    U = \{\epsilon\} \cup  \{a_1\dots a_j \mid j\in [1,k]\}.
\]%
If $v,u \in U$ and $v\neq u$ then $v\not\equiv_{\hat{L}} u$, since there is a
$\sigma \in \Sigma^*$ such that $v\sigma = w$ and $w\in \hat{L}$ and $u\sigma
\not \in \hat{L}$. This means there are $|U| = k+1$ distinct classes of
${\equiv_{\hat{L}}}$. Lastly, since $\sharp z \not\in \hat{L}$ for any
$z\in\Sigma^*$ we can also conclude that $\sharp \not\equiv_{\hat{L}} u$ for all
$u\in U$. 

Since we assumed that $A_{dist}$ has at most $k+2$ states, by
Corollary~\ref{cor:myhill-nerode} there are at most $k+2$ equivalence classes of
${\equiv_{\hat{L}}}$. Since trivially $w_1\sharp \not\equiv_{\hat{L}} \sharp$,
by the pigeonhole principle there is a prefix $u \in U$ such that at $w_1
\sharp \equiv_{\hat{L}} u$.

It can not be the case that $u = a_1 \dots a_i$ for some $i\in [1,k]$, since 
\begin{alignat*}{2}
    a_1\dots a_i &\cdot a_{i+1} \dots a_k \sharp w_{post} \in \hat{L}\\
    w_1\sharp &\cdot a_{i+1}\dots a_k \sharp w_{post} \not\in \hat{L}.
\end{alignat*}%
By eliminating all alternatives we conclude $u=\epsilon$. Using this equivalence
 and since $\epsilon \cdot w_1\sharp w_{post} \in L_{dist}$ we derive that
 $w_1\sharp \cdot w_1\sharp w_{post} \in L_{dist}$. In particular, this means
 that $(w_1\sharp)^n \cdot w_{post} \in L_{dist}$. By definition of $L_{\phi}^+$
 this means that the truth assignment $w_1$ satisfies all clauses $C_1, \dots,
 C_n$ and hence it is a satisfying assignment for $\phi$. This witnesses that
 $\phi$ is a satisfying formula. \qedhere 
\end{description}
\end{proof}

\noindent This lemma allows us to prove Theorem~\ref{thm:np-complete}.
\begin{proof}[Proof of Theorem~\ref{thm:np-complete}]
    Membership of NP follows naturally. For two DFAs $A_1$ and $A_2$ we can, in
    polynomial time, check if $\Lang(A_1) \subseteq \Lang(A_2)$. This can be
    done by computing the emptiness of $\Lang(A_1) \cap \overline{\Lang(A_2)}$.
    Moreover, either $A_1$ or $A_2$ itself necessarily already is a
    distinguishing automaton, so the minimal distinguishing DFA is definitely
    polynomial in size.

    NP-hardness is a direct consequence of Lemma~\ref{lemma:subseteq} and of the
    fact that $L_{\phi}^{-} \subseteq L_{\phi}^{+}$, so the language of any
    distinguishing automaton is a subset of $L_{\phi}^{+}$ and not vice-versa. 
\end{proof}

\paragraph*{Acknowledgements:}  The author thanks Tim Willemse for raising the
question of distinguishing transition-systems with invariants. Thanks also to
Jan Friso Groote and Anna Stramaglia for providing helpful suggestions on this
document.

\bibliographystyle{plain}
\bibliography{references}

\begin{thebibliography}{1}

\bibitem{demaine2011remarks}
Erik~D Demaine, Sarah Eisenstat, Jeffrey Shallit, and David~A Wilson.
\newblock Remarks on separating words.
\newblock In {\em Descriptional Complexity of Formal Systems: 13th International Workshop, DCFS 2011, Gie{\ss}en/Limburg, Germany, July 25-27, 2011. Proceedings 13}, pages 147--157. Springer, 2011.

\bibitem{gold1978}
E~Mark Gold.
\newblock Complexity of automaton identification from given data.
\newblock {\em Information and Control}, 37(3):302--320, 1978.

\bibitem{hopcroft1971DFAmin}
John Hopcroft.
\newblock An n log n algorithm for minimizing states in a finite automaton.
\newblock In Zvi Kohavi and Azaria Paz, editors, {\em Theory of Machines and Computations}, pages 189--196. Academic Press, 1971.

\bibitem{kozen1977lower}
Dexter Kozen.
\newblock Lower bounds for natural proof systems.
\newblock In {\em 18th Annual Symposium on Foundations of Computer Science (sfcs 1977)}, pages 254--266. IEEE, 1977.

\bibitem{kupferman2022certifying}
Orna Kupferman, Nir Lavee, and Salomon Sickert.
\newblock Certifying dfa bounds for recognition and separation.
\newblock {\em Innovations in Systems and Software Engineering}, 18(3):405--416, 2022.

\bibitem{moore1956gedanken}
Tyler Moore.
\newblock Gedanken-experiments on sequential machines.
\newblock In C.~E. Shannon and J.~McCarthy, editors, {\em Automata Studies, Annals of Mathematical Studies, no. 34}. Citeseer, 1956.

\bibitem{pfleeger1973state}
Charles~P Pfleeger.
\newblock State reduction in incompletely specified finite-state machines.
\newblock {\em IEEE Transactions on Computers}, 100(12):1099--1102, 1973.

\bibitem{smetsers2016minimal}
Rick Smetsers, Joshua Moerman, and David~N Jansen.
\newblock Minimal separating sequences for all pairs of states.
\newblock In Adrian-Horia Dediu, Jan Janou{\v{s}}ek, Carlos Mart{\'i}n-Vide, and Bianca Truthe, editors, {\em Language and Automata Theory and Applications ({LATA} 2016)}, pages 181--193. Springer, 2016.

\end{thebibliography}
\end{document}